\documentclass{llncs}
\usepackage[utf8]{inputenc}
\usepackage{hyperref}
\usepackage{llncsdoc}
\usepackage[algoruled,longend,vlined]{algorithm2e}
\DontPrintSemicolon
\usepackage{graphicx}
\usepackage{amsmath}
\usepackage{amssymb}
\usepackage[font={small}]{caption, subfig}
\usepackage{egweblnk}
\usepackage[textsize=footwww.google.notesize]{todonotes}
\usepackage{paralist}
\defaultleftmargin{0.0em}{0.0em}{0.0em}{0.0em}
 \usepackage{times} 
\usepackage[left, pagewise, mathlines]{lineno}
\begin{document}
\title{GraphMaps: Browsing Large Graphs\\ as Interactive Maps}
\author{Lev Nachmanson\inst{1} \and Roman Prutkin\inst{2}  \and Bongshin Lee\inst{1} \and
Nathalie Henry Riche\inst{1} \and Alexander E.~Holroyd\inst{1} \and Xiaoji Chen\inst{3}}
\institute{Microsoft Research, Redmond, WA, USA \\
  \texttt{\{levnach,bongshin,nath,holroyd\}@microsoft.com}
\and
Karlsruhe Institute of Technology, Germany \\
\texttt{roman.prutkin@kit.edu}
\and 
Microsoft, Redmond, WA, USA \\
  \texttt{missx@xbox.com}
}
  
\newcommand{\N}{\mathbb{N}_0}
\newcommand{\rephrase}[3]{\noindent\textbf{#1~#2.}~\emph{#3}}

\newcommand{\lev}[1]{\todo[color=yellow!40]{LN: #1}}
\newcommand{\romanp}[1]{\todo[color=green!40]{RP: #1}}

\maketitle
\begin{abstract}
Algorithms for laying out large graphs have seen significant progress in
the past decade. However, browsing large graphs remains a
challenge. Rendering thousands of graphical elements at once often
results in a cluttered image, and navigating these elements naively
can cause disorientation. To address this challenge we propose a
method called GraphMaps, mimicking the browsing experience of online
geographic maps.

GraphMaps creates a sequence of layers, where each layer refines the
previous one. During graph browsing, GraphMaps chooses the layer
corresponding to the zoom level, and renders only those entities of the
layer that intersect the current viewport. The result is that,
regardless of the graph size, the number of entities rendered at each
view does not exceed a predefined threshold, yet all graph
elements can be explored by the standard zoom and pan operations.

GraphMaps preprocesses a graph in such a way that during browsing, the
geometry of the entities is stable, and the viewer is responsive. Our
case studies indicate that GraphMaps is useful in gaining an overview of a
large graph, and also in exploring a graph on a finer level of
detail.
\end{abstract}
\section{Introduction}
Graphs are ubiquitous in many different domains such as information
technology, social analysis or biology. Graphs are routinely
visualized, but their large size is often a barrier. The difficulty
comes not from the layout which can be calculated very fast. (For
example, by using Brandes and Pich's algorithm~\cite{UlrikPich} a
graph with several thousand nodes and links can be laid out in a few
seconds on a regular personal computer.)  Rather, viewing and browsing
these large graphs is problematic. Firstly, rendering thousands of
graphical elements on a computer might take a considerable time and
may result in a cluttered image if the graph is dense.  Secondly,
navigating thousands of elements rendered naively disorients the user.

\newcommand*{\shifttext}[2]{%
  \settowidth{\@tempdima}{#2}%
  \makebox[\@tempdima]{\hspace*{#1}#2}%
}
\begin{figure*}[tcb]
  \centering
\begin{picture}(340,252)
\put(0,102){\includegraphics[width=.49 \linewidth]{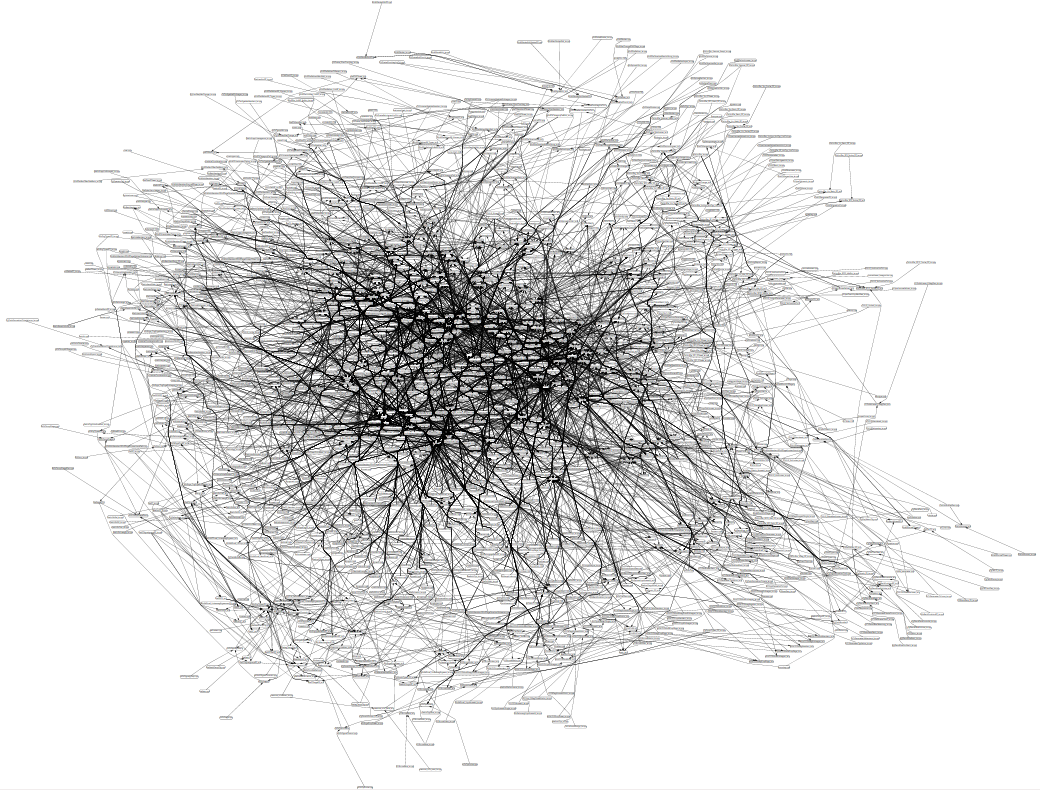}}
\put(5,105){{\setlength{\fboxsep}{0pt}\colorbox{white}{(a)}}} 
 \put(172,102){\includegraphics[width=.49 \linewidth]{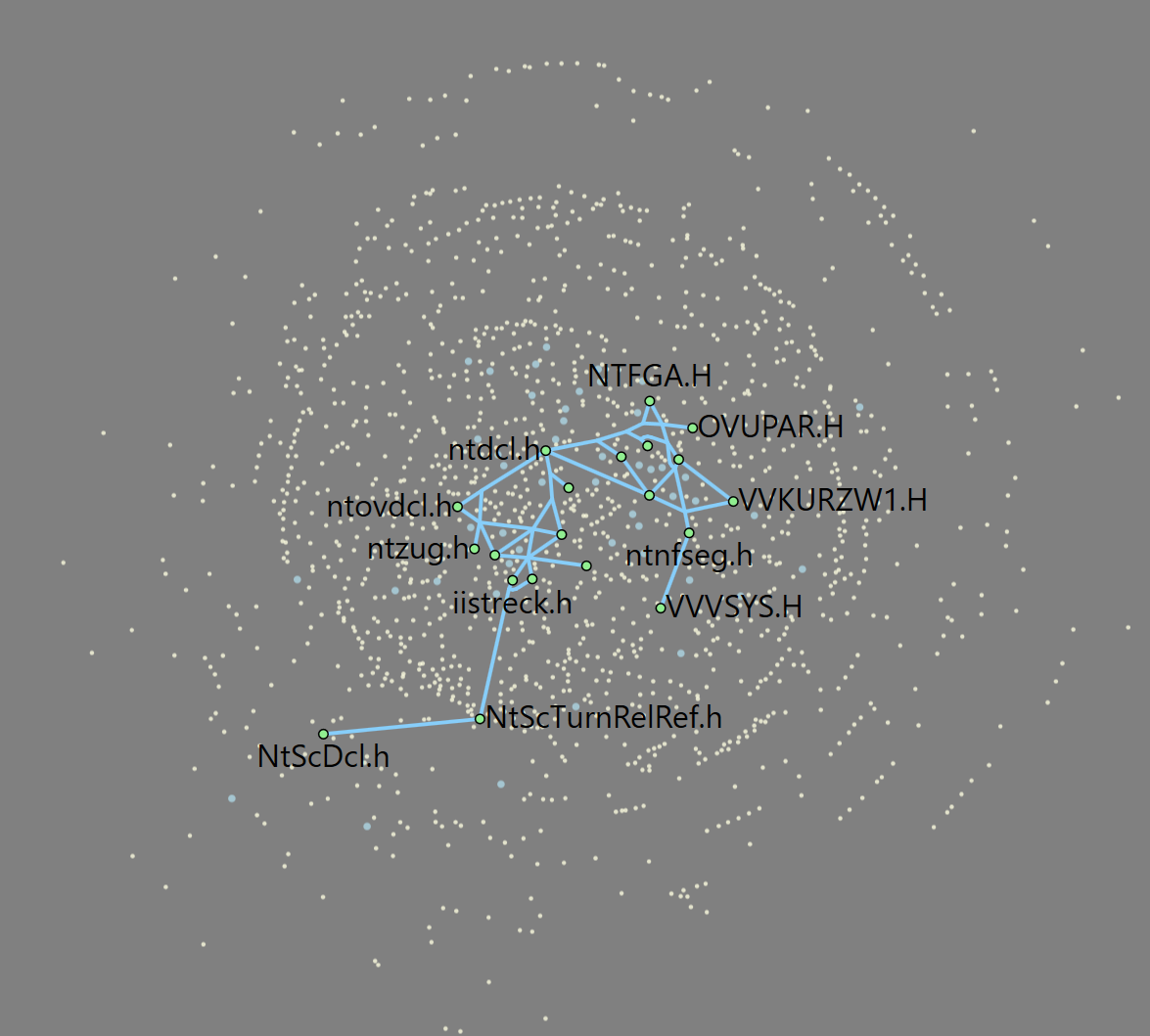}}
\put(175,105){{\setlength{\fboxsep}{0pt}\colorbox{white}{(b)}}} 
\hfill
 \put(0,0){\includegraphics[width=.49 \linewidth]{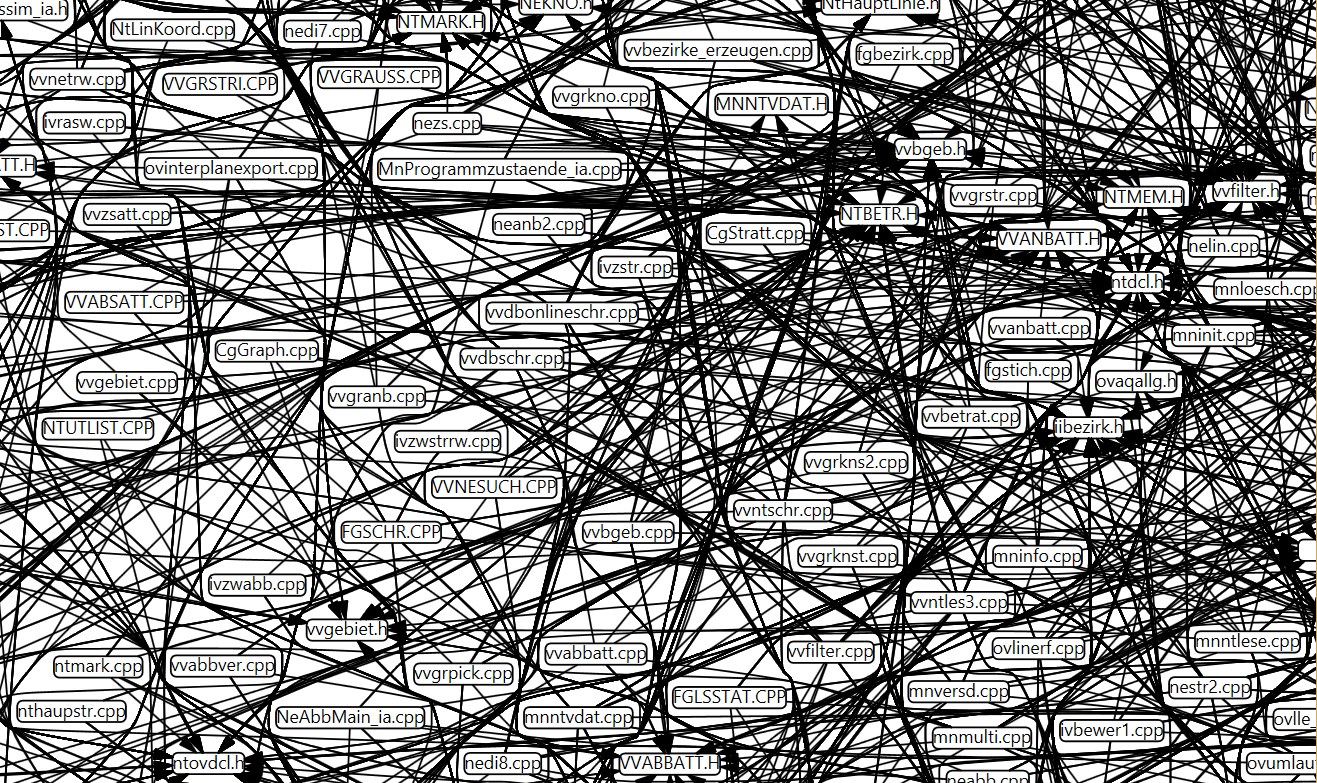}}
\put(5,3){{\setlength{\fboxsep}{0pt}\colorbox{white}{(c)}}} 
 \put(172,0){\includegraphics[width=.49 \linewidth]{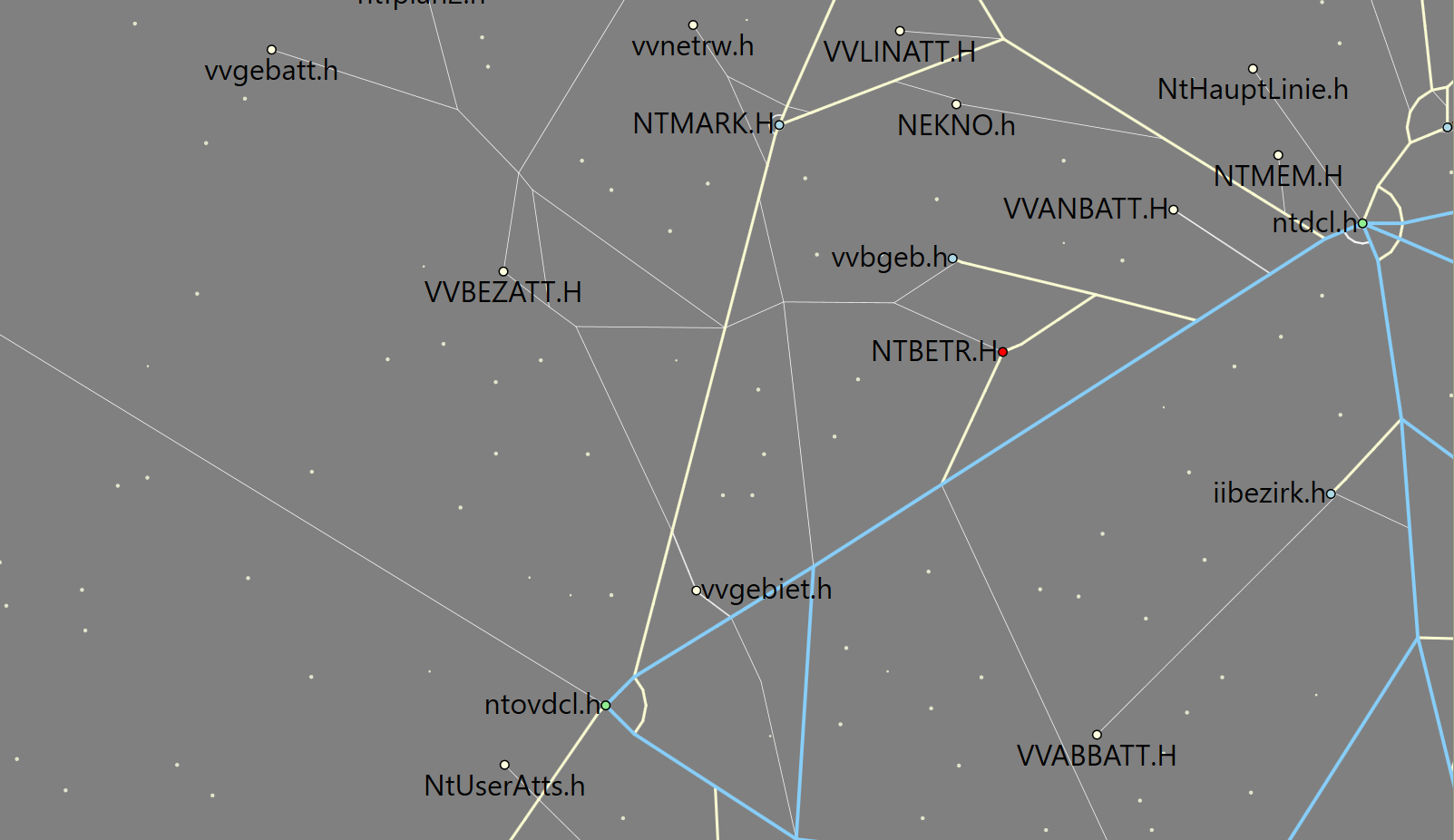}}
\put(175,3){{\setlength{\fboxsep}{0pt}\colorbox{white}{(d)}}} 
\end{picture} 
  \caption{\label{fig:b100} A graph\protect\footnotemark~with $1436$
    nodes and $5806$ edges. (a)~The full view with a standard method
    which draws all nodes and edges regardless of the zoom. (b)~The
    full view rendered by GraphMaps. (c)~A view with the zoom close to
    $9.13$ with the standard viewing. (d)~A view with zoom 9.26 with
    GraphMaps.}
\end{figure*}
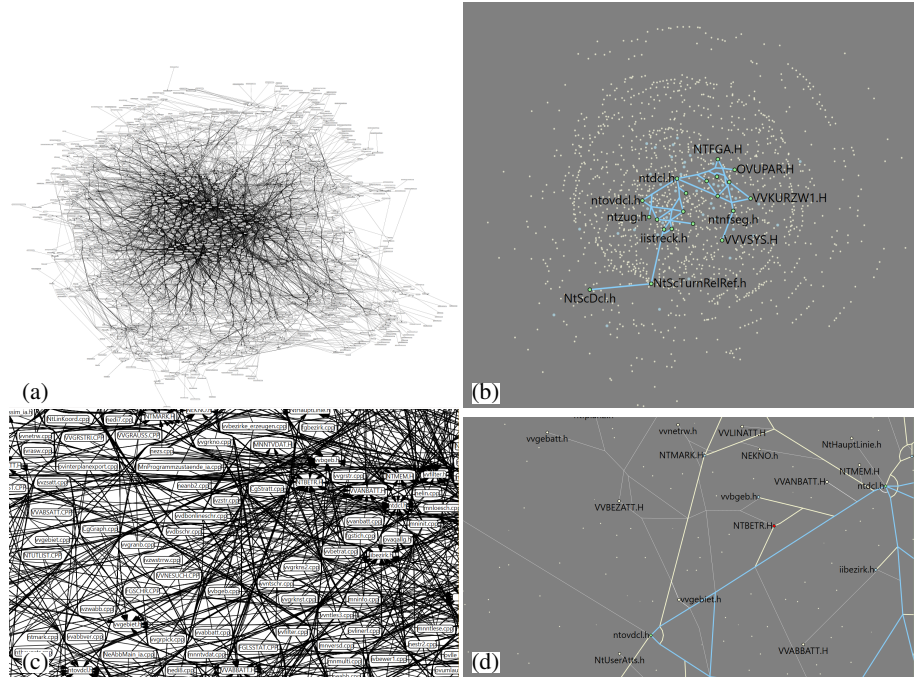

Our intention is to provide a graph browsing experience similar to
that of online geographic maps, for example, Bing or Google Maps. We
propose a set of requirements for such a visualization and introduce a
method, GraphMaps, fulfilling these requirements. GraphMaps renders a
graph as an interactive map by displaying only the most essential
elements for the current view. We allow fast interactions using
standard pan and zoom operations.  The drawing is visually stable, in
the sense that during these operations, nodes do not change their
relative positions, and edges do not change their geometry. To the
best of our knowledge, GraphMaps is the first method having these
properties. Fig.~\ref{fig:b100} illustrates the
method\footnotetext{https://github.com/ekoontz/graphviz/blob/master/rtest/graphs/b100.dot}.

\subsubsection*{Related Work.}
The problem of visualizing large graphs has been extensively addressed
in the literature, but here we discuss only the approaches most relevant to ours.
%
Most research efforts have concentrated on reducing the number of
visual elements to make node-link diagrams readable. We mention three
different approaches.

\paragraph{Aggregation techniques}

group vertices and edges of the graph together to obtain a smaller
graph~\cite{GenGeomGr}.  Most techniques compute a hierarchical
partitioning and offer interaction to explore different branches of
the tree. Early work by Eades and Feng~\cite{eades97} proposes
3-dimension visualization to navigate in this tree. Later
research~\cite{abello2006ask} demonstrated that techniques
implementing this approach can scale up to very large graphs ($16$
million edges and $200,000$ nodes).
Similar approaches attempt to give more clues about the content of the
aggregates.  Balzer and Deussen~\cite{balzer2007level} represent
aggregates by 3-dimensional shapes, whose sizes convey the number of
vertices, with bundled edges whose thickness indicates the density of
the connection. Zinmaier et al.~\cite{zinsmaier2012interactive}
utilize the GPU to create an aggregated image of a large graph, using
heatmaps to convey the number of vertices and edges in the aggregates.


While these techniques can scale up to very large graphs, they have
several disadvantages.  Aggregating nodes involves a loss of
information concerning intra- and inter-connectivity. Spatial
stability is another issue. The drawing may change dramatically when
several entities collapse into one, potentially disorienting the
user.

\paragraph{Multiscale techniques} allow users to explore the partition hierarchy
at different depths. These techniques aim at disambiguating the
topology induced by aggregating vertices and edges together. Auber et
al.~\cite{auber2003} propose a clustered multiscale technique, for
which the interiors of the aggregates are shown at a finer scale.
However, aggregated edges are shown between clusters, risking
misinterpretation. Henry et al.~\cite{nodetrix} propose a hybrid
technique that can only represent one level of clustering.  In a
similar spirit van Ham and van Wijk~\cite{smallworld2004} propose an
aggregate method in which users can expand one aggregate at a time.
Henry et al.~\cite{duplications} attempt to indicate inter-aggregate
connectivity by duplicating elements, but their solution only works
for a single level of clustering.


A different technique by Koren et al.~\cite{KorenTopologicalFish} aims
at smoothly integrating the level of detail, as opposed to discrete
partitioning of the graph.  The authors build a hierarchy of graphs
and, for each viewpoint, construct a smaller graph by ``borrowing''
parts of the corresponding hierarchy levels and adjusting the layout
of this smaller graph.  The strength of this technique is that it
avoids potentially misleading partitioning of the graph. However,
there is a lack of stability: a small change in viewport may lead to a
large change in the viewed graph. The fisheye technique of the paper
may also add a spatial distortion, further disrupting the user's
mental map.

\paragraph{Filtering techniques} approach the visualization of large
graphs by filtering the elements rendered in the view. For example,
SocialAction~\cite{socialaction} provides a set of measures to rank
vertices and edges, rendering node-link diagrams with manageable
sizes. A related technique by Perer and van Ham~\cite{van2009search}
proposes to build a filtered node-link diagram based on the queries
made by the user, via the concept of degree-of-interest. The principal
disadvantage of these techniques is the lack of overview of the entire
graph.
The progressive rendering approach proposed by Auber et
al.~\cite{Auber2004,Auber2002} renders the node-link diagram entities
in order of their importance. The rendering stops when the view
changes. Given enough non-interaction time, all entities intersecting
the viewport are rendered. In contrast to the previous filtering
techniques, the benefit of this approach is to reveal the key features
of the graph first. However, the user does not directly control the
level of detail, which potentially disrupts the experience.
\subsubsection{Design Rationale Motivated by Online Maps.}

Exploring online geographic maps is probably the most common scenario
for browsing large graphs. Millions of people every day browse maps on
their cellular phones or computers for finding a location or driving
directions. We decided to search for key ideas used in interactive
geographic maps that could be applied to browsing general graphs.

One insight is that showing everything at all times is
counterproductive. In a digital map on the top level we only see major
cities and major roads connecting them. Objects on finer levels of
detail, like smaller roads, are not shown explicitly. They may be
hinted by using, for example, pre-rendered bitmap tiles. When we zoom
in, other, less significant features appear and become labeled.
%
%
Online maps can answer search queries such as finding a route from
source to destination or showing a point of interest close to the
mouse position.

\paragraph*{Design goals} identified are as follows:
\begin{compactenum} 
\item The method should be able to reveal most details of the graph by
  using only the zoom in, zoom out, and pan operations. As we zoom in,
  more vertices and edges should appear according to their
  importance. Interactions such as node or edge highlighting or search
  by label should help discover further details.
\item During these operations, the user's mental map must be
  preserved. In particular, vertex positions and edge trajectories
  should not change between zoom levels.
\item In order to limit visual clutter, the number of rendered visual
  elements at each view should not exceed some predefined bound.
\end{compactenum}

\section{Method Description}
\label{sec:method-description}

The input to the algorithm is a graph with given node positions; the
edge routes are not part of the input.
 The output is a set of
\emph{layers} containing nodes and edge routes. Let $G=(V,E)$ be the input
graph, where $V$ is the set of nodes and $E$ the set of edges.
The input also includes an ordering of $V$. This ordering should
reflect the relative importance of the vertices. If such an ordering
is not provided then we can sort the nodes, for example, by
using PageRank~\cite{page1999pagerank}, by node degree, or by
shortest-path betweenness~\cite{Brandes2007}. Finding a good order
reflecting the node importance is a separate problem which is outside
the scope of this research. Here we look at the node order as input
and consider $V =[v_1,\dots,v_N]$ to be an array.

Before giving a detailed description of the algorithm we describe its
high level steps.

We build the \textbf{layer}~0, denoted by~$L_0$, as follows.
For some number $k_0>0$ we assign nodes $v_1,\dots,v_{k_0}$ to~$L_0$
and route all edges $(v_m,v_n) \in E$ with $m,n \le k_0$.  Suppose we
have already built~$L_{i-1}$ containing vertices $v_j$, for~$j \le
k_{i-1}$. Then, if $k_{i-1} < N$, that is we have vertices that are
not assigned to a layer yet, for a number $k_i \geq k_{i-1}$ we assign
nodes $v_1,\dots,v_{k_i}$ to~$L_i$ and route all edges $(v_m,v_n) \in
E$ with $m,n \le k_i$. Otherwise we are done. Note that a node can be
assigned to several consecutive layers.
To achieve the assignment we define a function $z$ from $V$ to the set
$\{2^0,2^1,2^2,\dots\}$. The value $z(v)$ we call the \textbf{zoom
  level} of the node. For $n \in \mathbb N_0$, the layer $L_n$
contains node $v$ if and only if $z(v) \le 2^{n}$. For each layer an
edge is represented by a set of straight line segments called
\textbf{rails}. We define function $z$ on rails too, but the layer
assignment rule is different for rails; a rail $r$ belongs to $L_n$
iff $z(r)$ is equal to $2^n$.

\begin{figure}[tbh!]
\subfloat{\includegraphics[trim=0mm 0.2mm 0mm 0.2mm, clip]{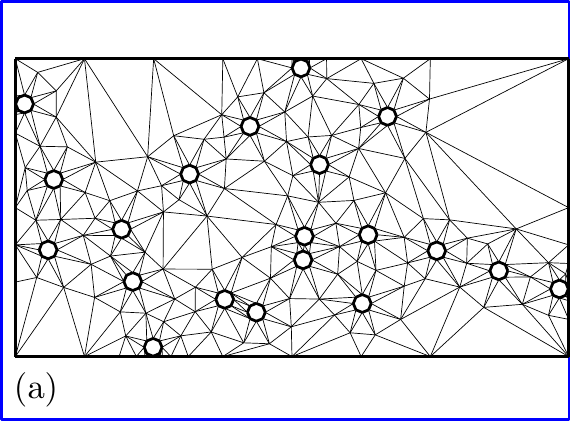}\label{fig:abstract:triang-0}}\hfill
\subfloat{\includegraphics[trim=0mm 0.2mm 0mm 0.2mm, clip]{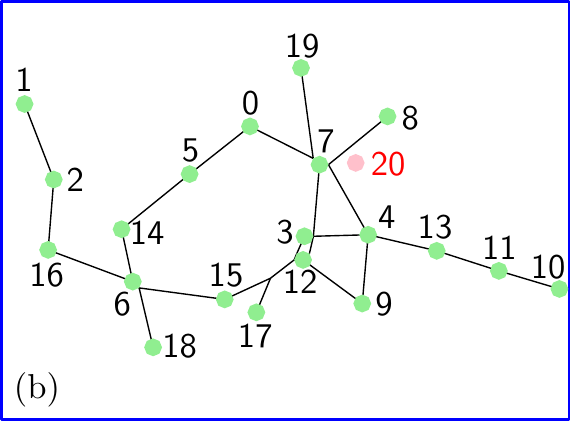}\label{fig:abstract:insertion-0}}\\
\subfloat{\includegraphics[trim=0mm 0.2mm 0mm 0.2mm, clip]{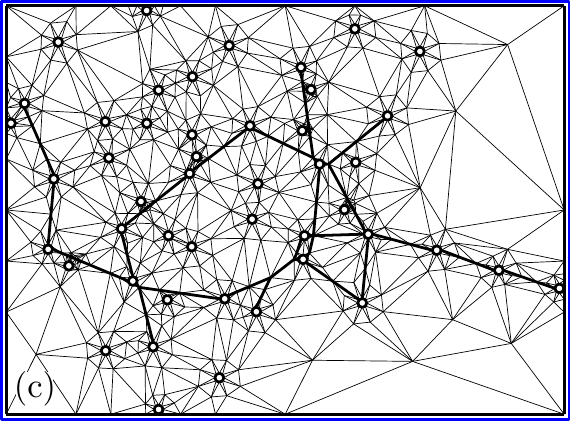}\label{fig:abstract:triang-1}}\hfill
\subfloat{\includegraphics[trim=0mm 0.2mm 0mm 0.2mm, clip]{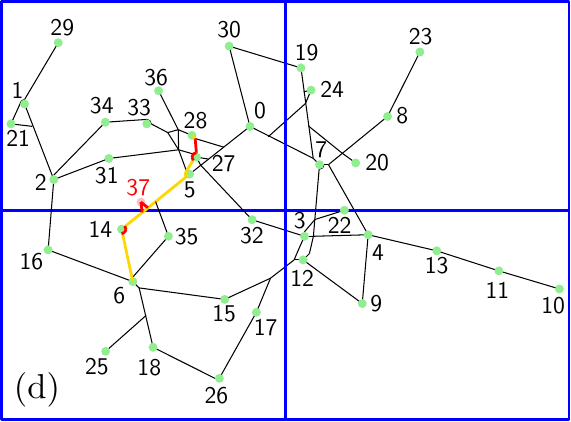}\label{fig:abstract:insertion-1}}\\
\subfloat{\includegraphics[trim=0mm 0.2mm 0mm 0.2mm, clip]{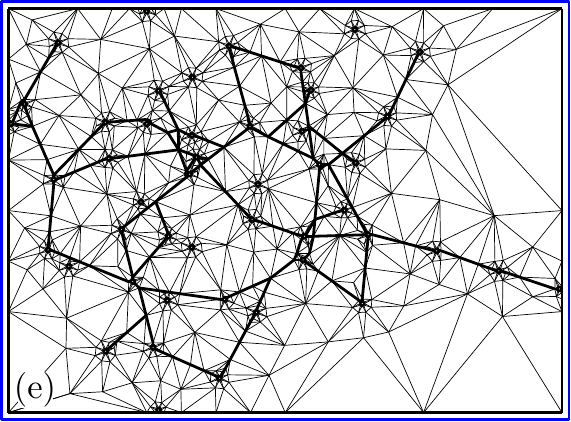}\label{fig:abstract:triang-2}}\hfill
\subfloat{\includegraphics[trim=0mm 0.2mm 0mm 0.2mm, clip]{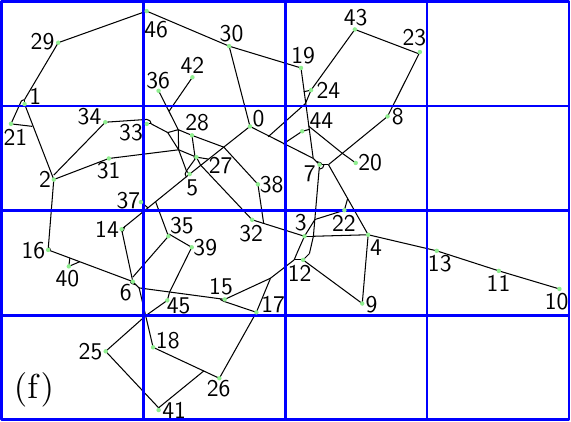}\label{fig:abstract:insertion-2}}
\caption{\label{fig:abstract} Graph \emph{abstract.dot}, $Q_N=80$,
  $Q_R=180$.  \protect\subref{fig:abstract:triang-0}~The mesh containing
  node boundaries of layer~0 (thick).
  \protect\subref{fig:abstract:insertion-0}~Nodes of layer~0 (green)
  and rails of edges routed using the mesh
  in~\protect\subref{fig:abstract:triang-0} (black). Adding node~20
  would exceed the node quota.
  \protect\subref{fig:abstract:triang-1}~The mesh containing rails and
  node boundaries of layer~0 and boundaries of candidate nodes for
  layer~1 (thick).
  \protect\subref{fig:abstract:insertion-1}~Node~37 has
  inserted nodes~6 and~28 as neighbors. The edges incident to Node~37 are routed
  through red rails, which are new maximal rails. After
  adding the red rails the upper left tile would intersect more
  than~$Q_R/4=45$ maximal rails.
  \protect\subref{fig:abstract:triang-2}~Mesh containing rails and
  node boundaries of layer~1 and candidates for
  layer~2.
  \protect\subref{fig:abstract:insertion-2}~All nodes and edges are
  added to layer~2 without exceeding the quotas. }
\end{figure}

\subsubsection{Calculation of layers.}

\newcommand{\blockline}{\noindent\hspace{-0.025\textwidth}%
    \rule[8.5pt]{0.989\linewidth}{0.8pt} \\[-0.80\baselineskip] }
\newcommand{\ProcessLayer}{\textbf{ProcessLayer}\xspace}
\newcommand{\maxRails}{\textbf{maximalRails}\xspace}
\newcommand{\maxRailsV}{\textbf{maximalRailsOfV}\xspace}
\newcommand{\insertedNodes}{\textbf{assignedNodes}\xspace}
\newcommand{\oldEdges}{\textbf{oldEdges}\xspace}
\newcommand{\oldRails}{\textbf{oldRails}\xspace}
\newcommand{\candidateNodes}{\textbf{candidateNodes}\xspace}
\newcommand{\prevRails}{\textbf{prevLayerRails}\xspace}
\newcommand{\prevRailsUpd}{\textbf{prevLayerRailsUpdated}\xspace}
\newcommand{\tileMap}{\textbf{tileMap}\xspace}
\newcommand{\tileMapij}{\textbf{tileMap}$(i,j)$\xspace}
\newcommand{\newRails}{\textbf{newRails}\xspace}
\newcommand{\rails}[1]{\textbf{rails}(#1)\xspace}
\LinesNumberedHidden
\begin{algorithm}[tbh!]
\caption{Setting node zoom levels with rail quota}
\label{algo:node-rail-zl}
\hspace{-0.25cm}\textbf{SetNodeRailZoomLevels()} \;
\blockline
\setcounter{AlgoLine}{0}
\ShowLn \insertedNodes = $\emptyset$, \maxRails = $\emptyset$\;
\ShowLn tileSize = BoundingBox($G$), nodeSize = InitialNodeSize, $n = 0$\;
\ShowLn \While{$|\textnormal\insertedNodes| < |V|$}{\label{alg:outerLoop}
\ShowLn \textbf{ProcessLayer()}\;\label{alg:processOnLevel}
\ShowLn tileSize = tileSize * $0.5$, nodeSize = nodeSize * $0.5$, $n = n + 1$ \;
}
\blockline
\hspace{-0.25cm}\textbf{ProcessLayer()}\;
\blockline
{
\addtocounter{AlgoLine}{-3}
\ShowLn initialize $\tileMap$ with \insertedNodes and \maxRails \; \label{alg:initTileMap}
\ShowLn \candidateNodes = TryAddingNodesUntilNodeQuotaFull($n$)\; \label{alg:candidateNodes}
\ShowLn \prevRails = rails of $L_{n-1}$, or $\emptyset$ if $n = 0$ \;
\ShowLn $M$ = GenerateMesh(\insertedNodes $\cup$ \candidateNodes, \prevRails)\; \label{alg:generateMesh}
\ShowLn \prevRailsUpd = SegmentsOfMeshOn($M$,\prevRails) \; \label{alg:prevRailsUpd}
\ShowLn $z(\prevRailsUpd) = 2^n$ \; \label{alg:setPrevRailsUpd}
\ShowLn \ForEach {$v \in $~\textnormal\candidateNodes}{ \label{alg:candidateLoop}
\ShowLn \rails{$v$} = RailsOnEdgeRoutes($v$, \insertedNodes,
$M$)\; \label{alg:routeEdges}
\ShowLn \maxRailsV = FindMaximalRails(\rails{$v$}) \; \label{alg:FindMaximalRails}
\ShowLn \lIf{adding all \maxRailsV to $L_n$ exceeds rail quota \label{alg:candidateIf}}{ \Return}
\ShowLn set~$z(v)=z(\rails{v})=  2^n$\;
\ShowLn update $\tileMap$ with $v$ and \maxRailsV \; \label{alg:candidateAdd} 
\ShowLn \maxRails = $\maxRails \cup \maxRailsV$ \;
\ShowLn$\insertedNodes = \insertedNodes \cup \{ v \}$
} 
}
\end{algorithm}

\newcommand{\viewRect}{P}
\newcommand{\width}{\textnormal{width}}
\newcommand{\height}{\textnormal{height}}

Algorithm~\ref{algo:node-rail-zl} computes the function $z$ on the
nodes and extends it to the rails. The flow of the algorithm is
illustrated in Fig.~\ref{fig:abstract}.

Let~$B$ be the bounding box of~$G$ with width~$w$ and
height~$h$. For $i$,$j$,$n \in \mathbb N_0$ we define $T_{ij}^n$ as
the rectangle with width~$w_n = w/2^n$, height~$h_n =
h/2^n$ and the bottom left corner with coordinates $x=u+i
\cdot w_n$ and $y=v+j \cdot h_n$, where $(u,v)$ is the left bottom
corner of $B$. We call $T_{ij}^n$ a \textbf{tile}.
The algorithm is driven by positive integers $Q_N$ and~$Q_R$, which we
call \textbf{node} and \textbf{rail quota}, respectively.  We say that
tile $T_{ij}^n$ \emph{exceeds node quota}~$Q_N$ if it intersects
more than~$Q_N/4$ nodes of layer~$n$.

To work with the rail quota $Q_R$ we need the following
definition. For a set of rails $R$ and a rail $r \in R$ we call $r$
\textbf{maximal in $R$} if $r$ is not a sub-segment of any other rail
in $R$.
During the algorithm we maintain the set of maximal rails among the
set of rails already assigned to layers and count intersections
between the tiles and the maximal rails only.
The union of all maximal rails will always form the same set of points
as the union of all rails created so far.
Tile~$T_{ij}^n$ \emph{exceeds rail quota} if it intersects more than
$Q_R/4$ rails which are maximal among all rails of layer~$n$ and
below. Assume both~$Q_N$ and~$Q_R$ are divisible by~4.

The outer loop of Algorithm~\ref{algo:node-rail-zl} in
line~\ref{alg:outerLoop} works as follows. Starting with $n=0$, each
call to \ProcessLayer in line~\ref{alg:processOnLevel} tries to
greedily assign the nodes to the current layer. Each such attempt
starts with the first unassigned node in~$V$. Procedure \ProcessLayer
terminates if adding the next node in~$V$ and its edges incident to
already assigned nodes would exceed the node or rail quota of some tile.

After calling \ProcessLayer tile dimensions and the node size become
twice smaller, and a new attempt starts for $n+1$ in
line~\ref{alg:processOnLevel}. The algorithm stops when all nodes are
assigned to a layer.
Fig.~\ref{fig:abstract} illustrates Algorithm~\ref{algo:node-rail-zl}
for graph
\emph{abstract.dot}\footnote{https://github.com/ekoontz/graphviz/blob/master/rtest/graphs/abstract.dot}
with 47 nodes, labeled from~0 to~46 according to their order in~$V$.

In line~\ref{alg:initTileMap} \textbf{tileMap} is a map from $\N^2$ to
$\N^2$. If for some $n$ we have $\tileMap(i,j)$ $=(r,k)$, then $r$ nodes in $L_n$
and~$k$ maximal rails intersect $T_{ij}^n$.

Let us now consider \ProcessLayer for $n=0$. For this case the domain
of \tileMap is $\{(0,0)\}$. The sets \insertedNodes and \maxRails are
empty, and there is only one tile~$T_{0,0}^0$, which has the size
of~$B$ (blue in Fig.~\ref{fig:abstract:triang-0}
and~\ref{fig:abstract:insertion-0}). After executing
line~\ref{alg:candidateNodes} the set \candidateNodes contains the
first~$Q_N/4$ nodes of~$V$ (green in
Fig.~\ref{fig:abstract:insertion-0}). The boundaries of these nodes
are represented by regular polygons (thick in
Fig.~\ref{fig:abstract:triang-0}) and used to generate a triangular
mesh~$M$.  The mesh is a constrained triangulation in a sense that any
straight line segment of the input can still be traced in $M$ although
it can be split into several segments. The edges with both endpoints
in \candidateNodes are routed on $M$.

In the~$i+1$-th iteration of the loop in line~\ref{alg:candidateLoop}
the algorithm tries to add node~$i$, while nodes~$0, \dots, i-1$ have
already been added to $L_0$, and $\tileMap(0,0)=(i,k)$, where $k \le
Q_R/4$ is the number of rails used by edges routed so far. All these
rails are maximal rails by construction.

In line~\ref{alg:routeEdges} the routes of edges from node~$i$ to
nodes~$0, \dots, i-1$ are computed as shortest paths on~$M$, and the
set \rails{$v$} is the set of all rails of these routes.  In
line~\ref{alg:FindMaximalRails} we find \maxRailsV, the rails from
\rails{$v$} which are maximal with respect to the set $\maxRails \cup
\rails{$v$}$. In the case of~$n=0$ they are all the rails of
\textbf{rails}$(v)\setminus \maxRails$. For~$n \geq 1$, these are the
rails from \rails{$v$} covered by no rail from \maxRails. In
Fig.~\ref{fig:abstract:insertion-1}, such maximal rails for node~37
are drawn red.

If~$T_{0,0}^0$ still contains no more than~$Q_R/4$ rails after
adding~\maxRailsV, then node~$i$ is added to $L_0$. Otherwise,
\ProcessLayer terminates. In Fig.~\ref{fig:abstract:insertion-0},
all~$Q_N/4=20$ candidate nodes and the rails on the corresponding
edges could be added to $L_0$.

The procedure works similarly for~$n \geq 1$. One notable difference
is that rails from~$L_{n-1}$ are passed as input to the mesh generator
in addition to the boundaries of the appropriate nodes in
line~\ref{alg:generateMesh}. For more details
 we refer to the proof of
Lemma~\ref{smalltiles} in the
appendix. Fig.~\ref{fig:abstract:triang-1}, \dots,
\ref{fig:abstract:insertion-2} show \ProcessLayer for
$n=1,2$.

\subsubsection{Using the layers during the visualization.}\label{sub:usinglayers}
Let $H$ be a rectangle. We denote by $w(H)$ the width of $H$ and by
$h(p)$ the height of $H$. Recall that $B$ is the bounding box
of~$G$. Then the \emph{zoom level of $H$ to $B$} is the value
$l(H)=\min \{ \frac {w(B)}{w(H)}, \frac {h(B)}{h(H)}\}$.

Let ~$K$ be the transformation matrix from the graph to the user
window $W$. Then the rectangle $P=K^{-1}(W)$, where $K^{-1}$ is the
inverse of $K$, is the current viewport.

To decide which elements of~$G$ are displayed to the user, we find the
\emph{zoom level} $Z = l(P)$ and set the layer index $n = \max (0,
{ \lfloor \log_2 Z \rfloor })$. Finally, the elements displayed to the
user are all the nodes and rails of layer $L_n$ intersecting $P$.
We show in the appendix in Theorem~\ref{th:nodecorrect} that, by
following this strategy, we render at most $Q_N$ nodes, and the
rendered rails can be exactly covered by at most $Q_R$ maximal rails.



\subsubsection{Edge routing and overlap removal.}
\label{sec:edge-routing}
Consider the nodes of $L_0$. To construct a graph on which the edges
are routed, we first create a regular polygon for each vertex. Then, we
generate a triangular mesh using the \emph{Triangle} mesh generator by
Shewchuk~\cite{Shewchuk2002}. By inserting additional vertices
Triangle creates meshes with a lower-bounded minimum angle, which
implies the upper-bounded vertex degree. Each edge between a pair of
$L_0$ nodes is then assigned the corresponding Euclidean shortest path
in the mesh, which is computed using the $A^*$ algorithm. Mesh segments
lying on such paths become rails of $L_0$ and the remaining mesh
segments are discarded.

We now proceed with edge routing for $L_n$ for $n \geq 1$. Consider
Procedure \ProcessLayer. Since the initial node placement did not take
edge trajectories into account, at the beginning of the procedure some
unassigned nodes might overlap the entities of $L_{n-1}$. We move
these nodes away from their initial positions to resolve these
overlaps, but this way we might create new overlaps with the nodes that are not
assigned to a layer yet.
%
%

The overlap removal process happens before
line~\ref{alg:candidateNodes}. We follow the metro map labeling method
of Wu et al.~\cite{Wu2011}. All line segments and bounding boxes of
fixed nodes are drawn on a monochromatic bitmap and the image is
\emph{dilated} by the diameter of a node on $L_n$. To define a
position for a candidate node $v$ at which it does not overlap already
placed nodes or rails, we find a free pixel~$p$ in the image, ideally
close to the initial location of $v$. We draw a
dilated~$v$ at~$p$ and proceed with the next candidate, etc.

To generate a graph for edge routing on $L_n$, we use the bounding
polygons of nodes from \candidateNodes and the nodes of $L_{n-1}$, and
the rails of $L_{n-1}$, as the input segments for
 Triangle. Already routed edges maintain their trajectories, while
edges incident to a node not belonging to $L_{n-1}$ are routed over
the triangulation created by Triangle in line~\ref{alg:routeEdges}. To
create the bundling effect by reusing existing rails, we
slightly reduce their weights during routing.

\subsubsection{Pre-rendered tiles.}

To help users gain spatial orientation, we hint the nodes which are
not yet visible at the current zoom level, but will appear if we zoom
in further. For this purpose we create and store on the disk the
images of some graph nodes and use them as the background.
The images are generated very fast and are loaded and unloaded
dynamically by a background thread to keep the visualization
responsive. For more details see Section~\ref{app:sec:tiles} of the
appendix.

\subsubsection{Interaction.}

We define several interactions in addition to the zoom and pan.
Clicking on a node (even if it is hinted, but not visible yet)
highlights all edges incident to it and unhides all adjacent
nodes. The highlighted elements are always shown regardless of the
zoom level.
Clicking on a rail highlights the most important edge passing through
it, and unhides the edge endpoints.
Additionally, nodes can be searched by substrings of their labels.


\section{Experiments}

\subsubsection{Visualizing and evaluating clusterings.}

The test graph of our first experiment is the \emph{Caltech}
graph, which was used by Nocaj et al.~\cite{Nocaj2014}. It is the graph of
Facebook friendships at California Institute of Technology from
September 2005 and contains 769 nodes and 16k edges~\cite{Traud2011}.
The nodes are labeled by class year and residence, or house, of the
corresponding student, and colored by the house. The label 0 denotes missing data.

\begin{figure}[tbh!]
\fbox{\includegraphics[width = .477\linewidth]{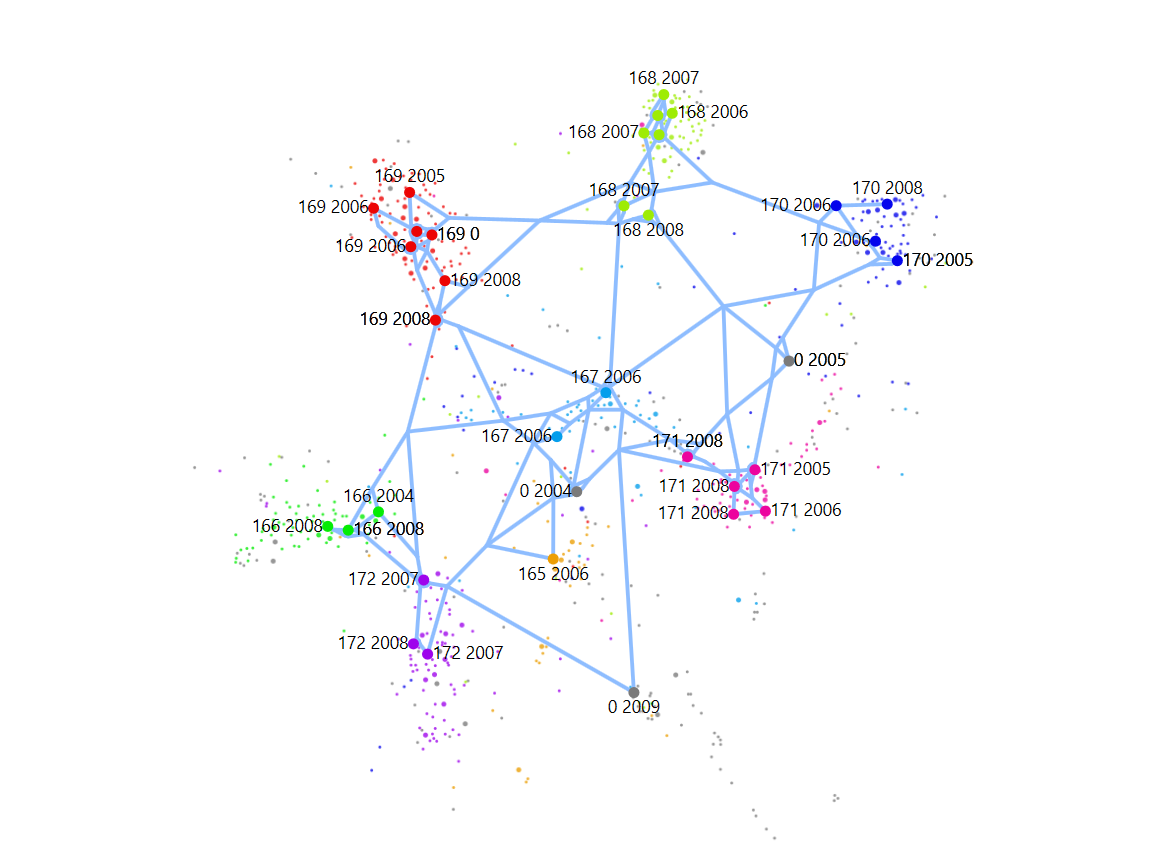}}
\fbox{\includegraphics[width = .477\linewidth]{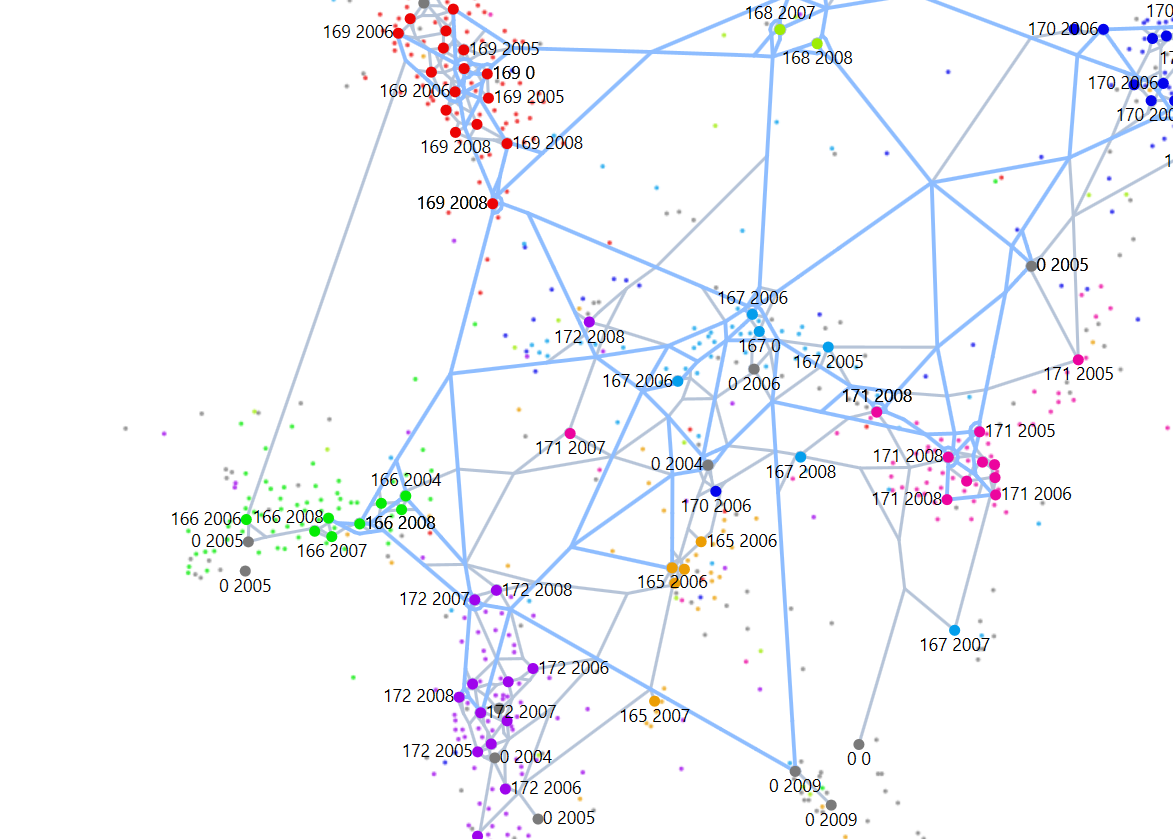}}
\fbox{\includegraphics[width = .477\linewidth]{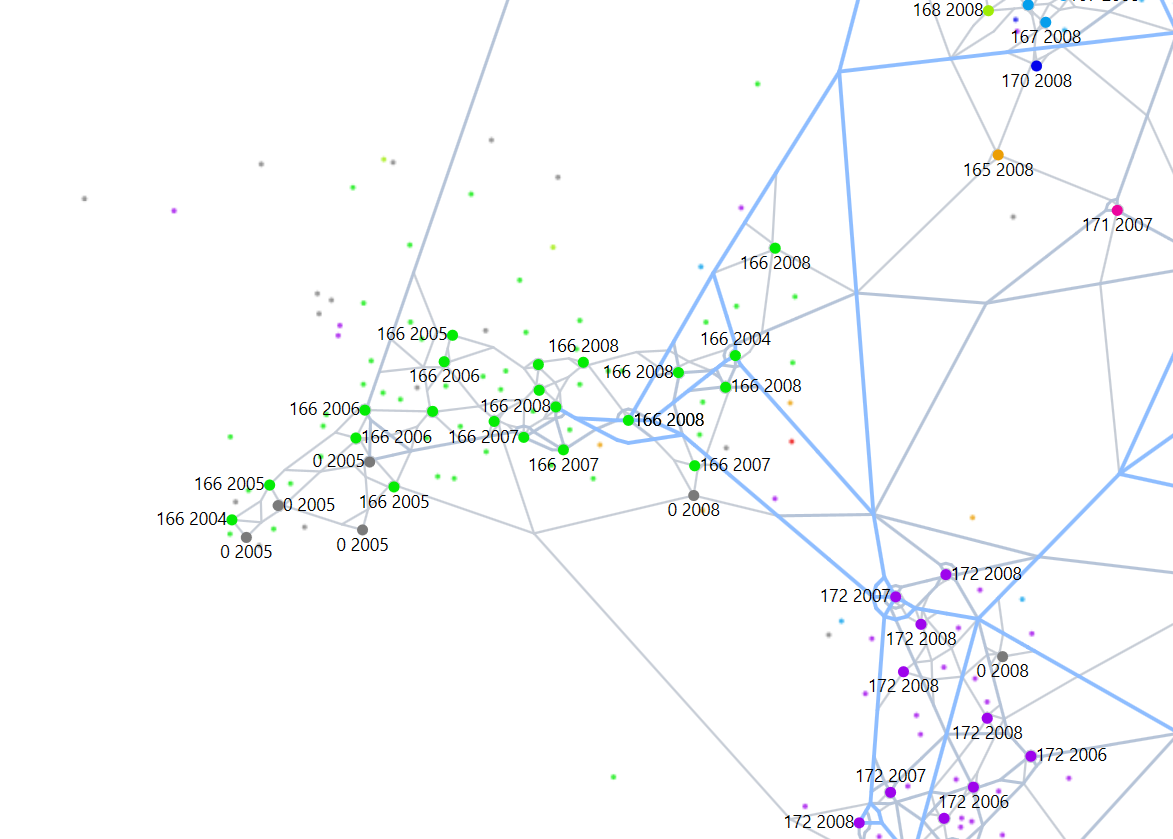}}
\fbox{\includegraphics[width = .477\linewidth]{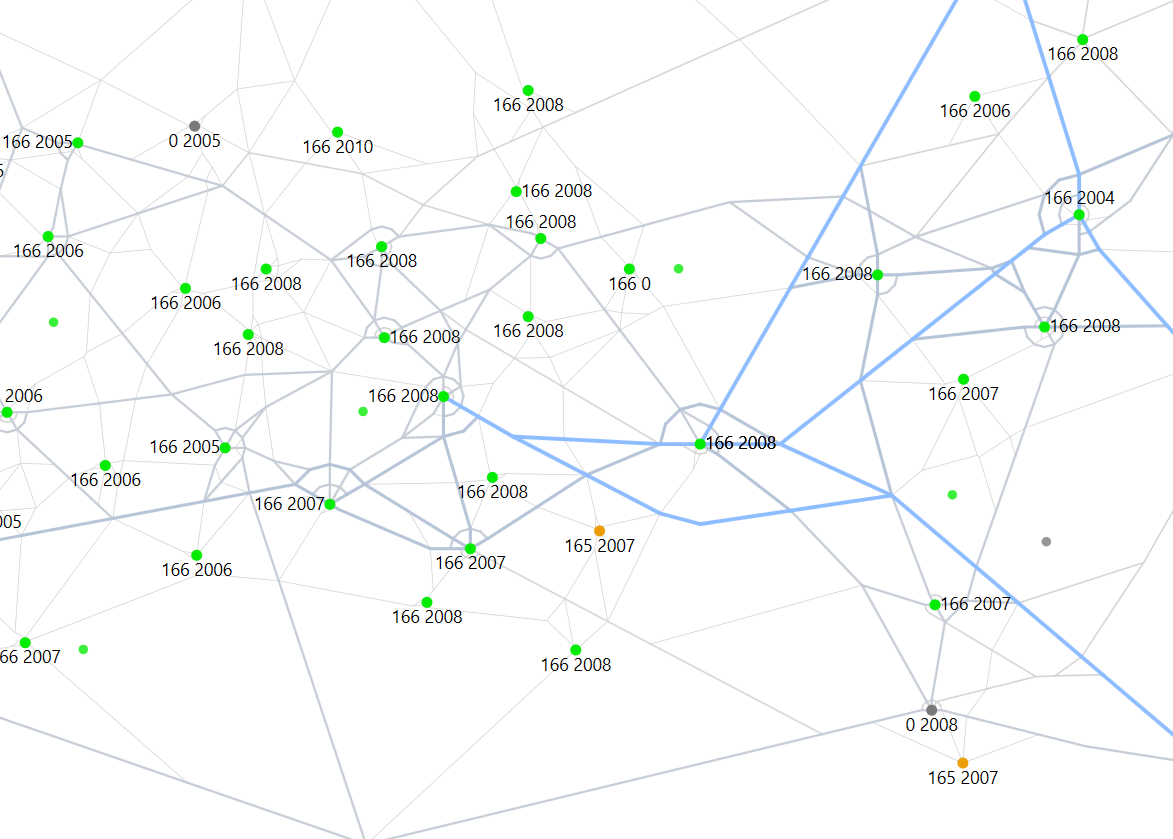}}
\caption{\label{fig:caltech:graphmaps} Caltech graph on four different
  zoom levels visualized using our approach.}
\end{figure}

A computer science researcher, with focus on clustering algorithms, used
GraphMaps to browse this graph. He was interested in discovering the
connectivity structure of the graph, e.g., which houses or years have
strong ties. The researcher's tool of choice for visualizing
clusterings was \emph{Gephi}~\cite{Gephi}. The node layout was
computed by a force-directed layout algorithm applied to the Simmelian
backbone of the graph, as proposed by Nocaj et al.~\cite{Nocaj2014}. %
The result is shown in Fig.~\ref{fig:gephi:all}. The same initial
layout was used as input for our algorithm. The resulting drawing on
four different zoom levels is shown in
Fig.~\ref{fig:caltech:graphmaps}.

\begin{figure}[tbh!]
\subfloat[]{\includegraphics[width = .33\linewidth]{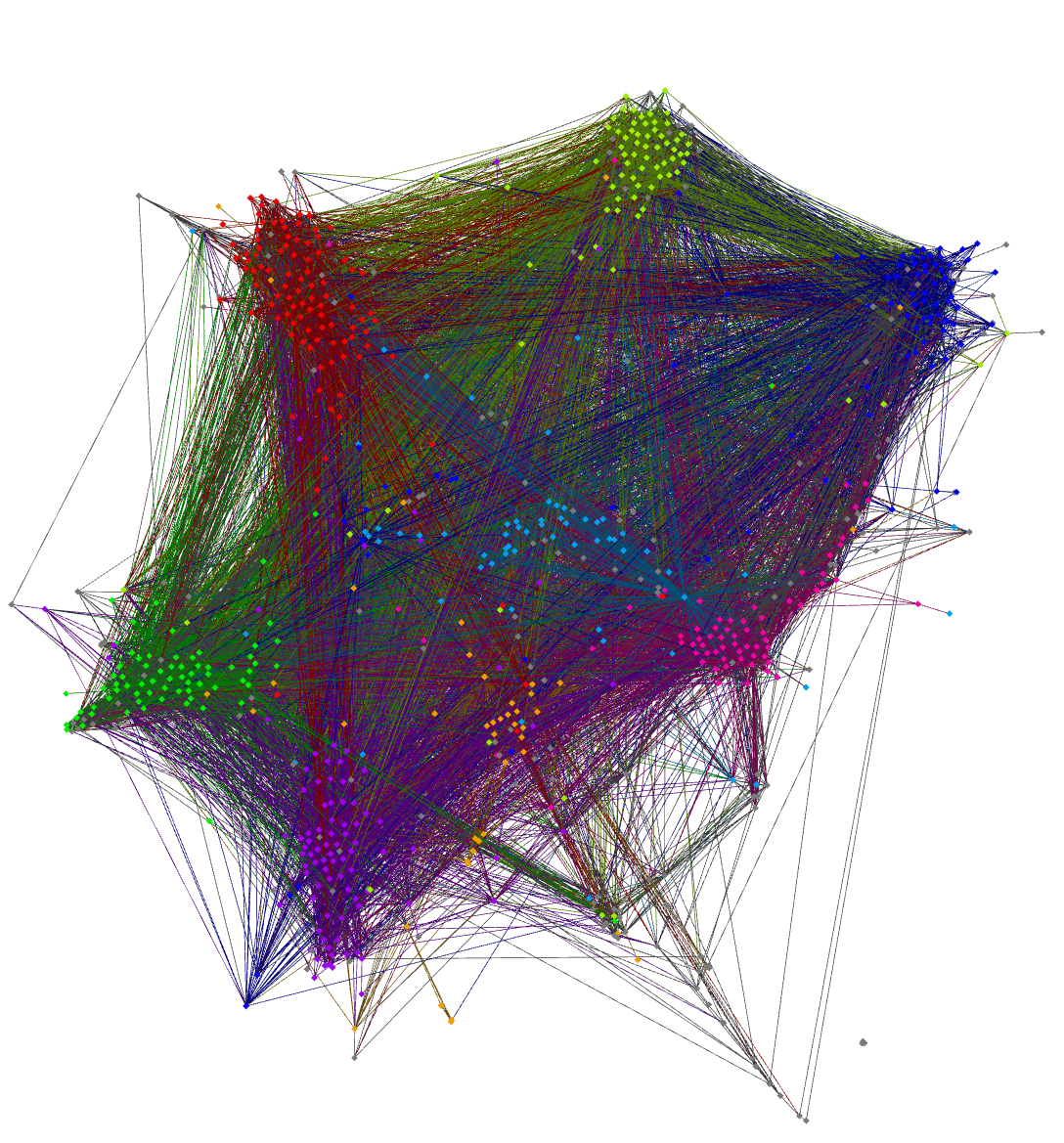}\label{fig:gephi:all}}
\subfloat[]{\includegraphics[width = .33\linewidth]{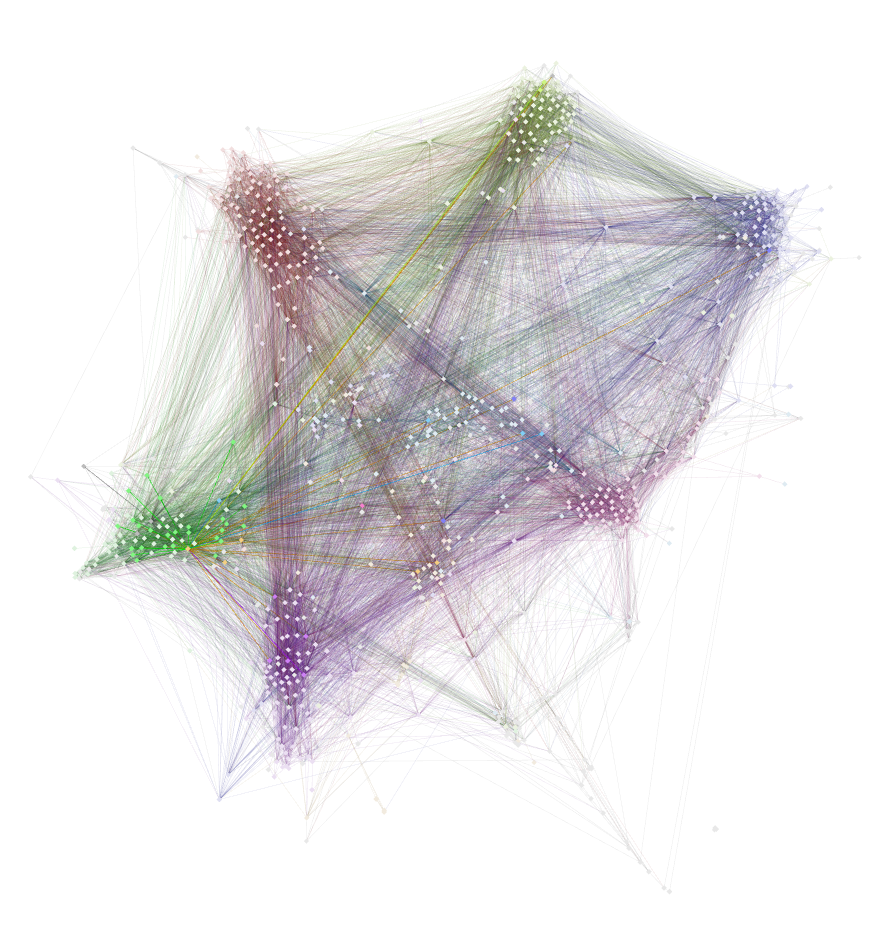}\label{fig:gephi:neighb}}
\subfloat[]{\includegraphics[trim = 50mm 20mm 50mm 15mm, clip, width = .3\linewidth]{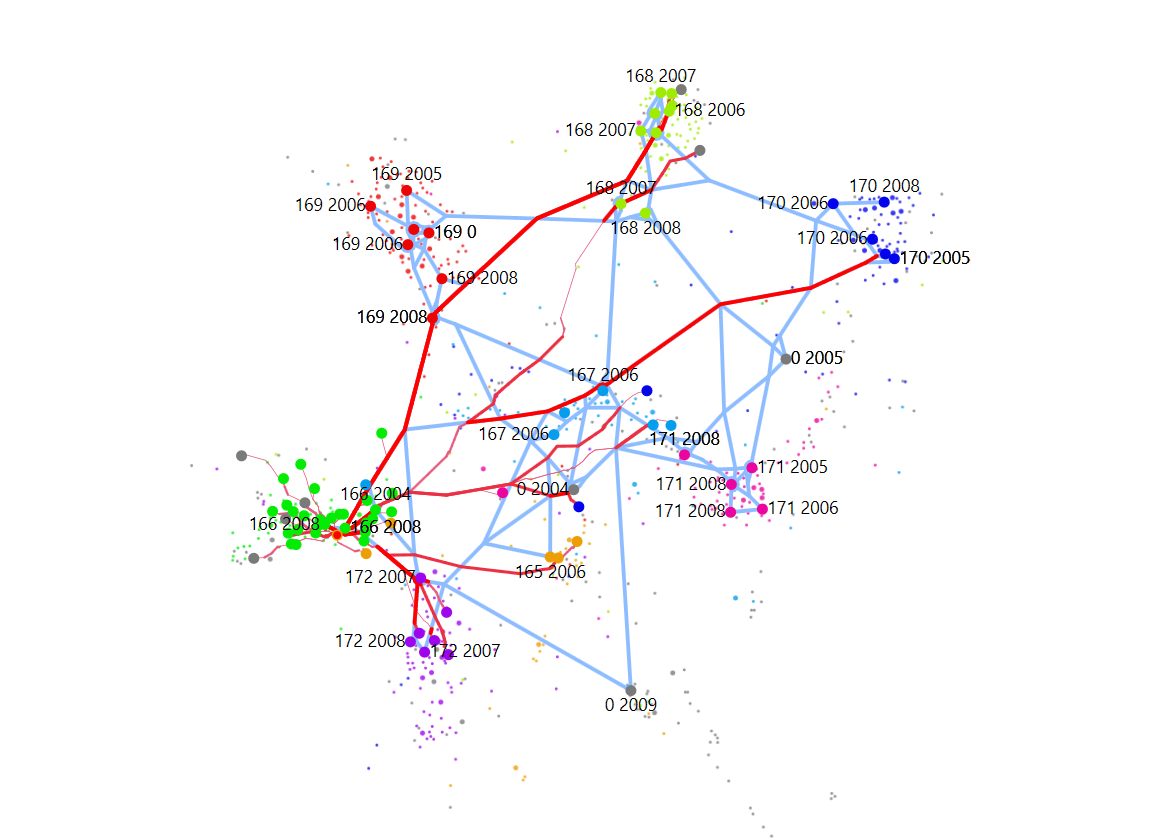}\label{fig:graphmaps:neighb}}
\caption{\label{fig:gephi} \protect\subref{fig:gephi:all}~Caltech
  graph visualized using Gephi. \protect\subref{fig:gephi:neighb},
  \protect\subref{fig:graphmaps:neighb}: showing neighbors of node
  ``165 2007'' (leftmost orange node surrounded by green nodes) using
  Gephi and our tool.}
\end{figure}

The user noted that the view in Fig.~\ref{fig:gephi:all} was too dense
and gave no insight into the graph connectivity. On the other hand, he
found our result in Fig.~\ref{fig:caltech:graphmaps} less
cluttered. User commented that, by looking at the edge routing created
by GraphMaps, one may think that two nodes are connected by an edge,
when, in fact, they are not. However, by using additional interactions
besides zoom and pan, e.g., edge highlighting, the connectivity can be
understood.

One interaction mode that the user tested for both tools was selecting
all neighbors of a node. In Gephi, when hovering the mouse over a
node, all non-incident edges and non-adjacent nodes are grayed out;
see Fig.~\ref{fig:gephi:neighb}. In our method, when clicking on a
node, routes of all its incident edges are highlighted and,
additionally, all adjacent nodes are shown, regardless of the zoom
level; see Fig.~\ref{fig:graphmaps:neighb}. According to the user,
both methods provided satisfactory results. He noted that GraphMaps, 
by using edge bundling, provides a tidier picture than Gephi.

For dense graphs, like Caltech, the user would prefer to view the
neighbors of a node in GraphMaps. The user commented that, contrary to Gephi,
GraphMaps exposes the most important nodes and their labels in a
readable fashion.
\subsubsection{Experiments with other graphs.}
In the video at {\small{\webLinkFont\httpAddr{//1drv.ms/1IsBEVh}}} we
demonstrate browsing the graph
``composers''\footnote{http://www.graphdrawing.de/contest2011/topic2-2011.html}
  with GraphMaps. The nodes of the graph represent the articles on
  Wikipedia on composers, and the edges represent Internet links
  between the articles. In the video we demonstrate the user
  interactions that help us to explore the graph. 

  When browsing the graph of InfoVis coauthors, created from ACM data,
  another user was able to notice two groups of coauthors, one
  connected to Peter Eades, and another one to Ulrik Brandes and
  Michael Goodrich. By selecting all direct neighbors of Peter Eades,
  the user was able to see that only one member of the second group,
  Roberto Tamassia, has a paper with Peter Eades; see
  Fig.~\ref{fig:groups} in the appendix. Further analysis showed that,
  according to the data set, Roberto Tamassia, is the only author with
  coauthors from both groups. GraphMaps enabled the user to gain
  insights on the graph structure.

  \subsubsection{Running time.}GraphMaps processes a graph with 1463
  nodes and 5806 edges for 1 minute, a graph with 3405 nodes and 13832
  edges for 130 seconds, and a graph with 38395 nodes and 85763 for
  less than 6 hours. The experiments were done on an HP-Z820 with Intel
  Xeon CPU E5-2690 under Windows 8.1. The required memory was 16
  GB. The current bottleneck in performance is the edge routing. We
  hope to speed up the edge routing by using parallel processing.
\subsubsection{The sources of GraphMaps.}
GraphMaps is implemented in MSAGL, which is available as Open Source
at {\small
  {\webLinkFont{github.com/Microsoft/automatic-graph-layout}}}. 
\section{Discussion}

The users of GraphMaps appreciate its
aesthetics and the similarity to browsing online maps. GraphMaps helps
in gaining the first impression of the graph structure and, in spite
of the fact that precise knowledge of the connectivity cannot be
obtained with GraphMaps by zooming and panning alone, additional
interactions allow answering the queries as, for example, finding if
two nodes are direct neighbors. A current shortcoming of GraphMaps is
that the direction of the edges is lost. It happens for other methods
as well, when edges are bundled. Solving this issue is a possible
future work item. The labeling algorithm needs improvement, since it
does not always respect the node ranking and does not always utilize free
space well enough.
\subsubsection{Future work.}
Currently we cannot guarantee that our layer generation algorithm
always reaches the end, although, in all our experiments it
did. Creating a version of the algorithm which provably stops, or,
even better, guarantees that the number of generated layers is within
predefined bounds, is a very interesting problem.

Another problem is finding a node placement working nicely with the
node ranking and suitable to assigning the nodes to the layers.
Ideally, such a layout algorithm is aware of the edge routing too, and
avoids the overlap removal step.
\subsubsection*{Acknowledgements.}
We are grateful to Roberto Sonnino for the useful discussions on the
rendering of the tile images in a background thread, and to Itzhak
Benenson for sharing with us his ideas on the visualization style.

{\small
 \bibliography{abbrv,lg}
 \bibliographystyle{titto-lncs-01}
}
\newpage
\section{Appendix}

\begin{figure}[tbh!]
\subfloat[]{\includegraphics[width = 1\linewidth]{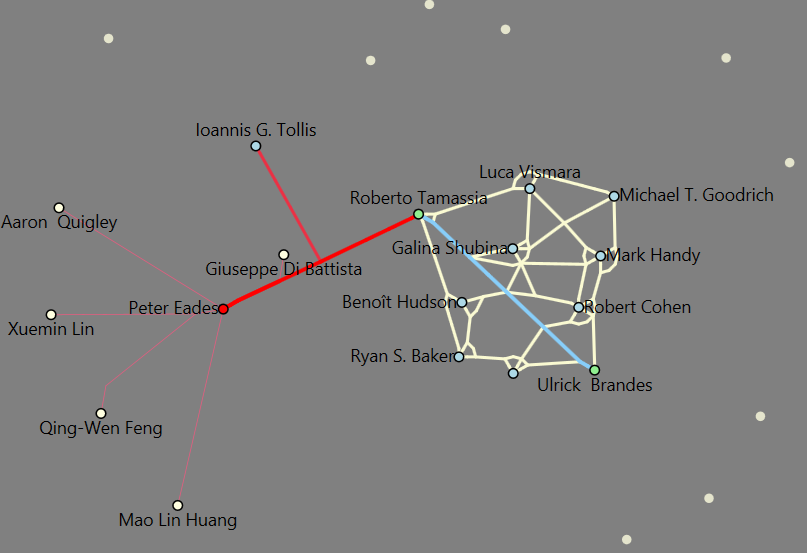}\label{fig:groups}}
\caption{We can establish that Roberto Tamassia is the only coauthor of
  Peter Eades in the right group}
\end{figure}

\subsection{Proof that quotas are satisfied}
%
%

Here we repeat the definitions from Section~\ref{sec:method-description}
for the sake of convenience. Let~$K$ be the current transformation
matrix from graph~$G$ to the user window rectangle~$W$, and let
rectangle $\viewRect = K^{-1}(W)$ denote the rectangle~$W$ mapped to
the graph coordinates. Let $B$ be the bounding box of $G$, and $w,h$
are the width and the height of a rectangle, respectively.
We then define the current zoom level $Z$ as $Z = l(P)$. Recall
that~$l(P)$ is defined as $l(P)=\min \{ \frac {w(B)}{w(P)}, \frac
{h(B)}{h(P)}\}$.
We choose $n = \max (0, { \lfloor \log_2 Z \rfloor })$ and render
elements of $L_n$ intersecting $\viewRect$. We need to prove that we
do not exceed the quotas when rendering. 

In general, let $H$ be any rectangle and let $n(H) = \max (0, {
  \lfloor \log_2 {l(H)} \rfloor })$. Let $V_H= \{v \in V, z(v) \le
l(H) \text{ and } v \text{ intersects } H\}$. Set $V_P$ represents all
the nodes that are rendered for the current view.  Let~$R$ be the set
of all rails on all layers. Let $R_H = \{r \in R: z(r) = l(H)\
\text{and } r \cap H \ne \emptyset \}$. In this notation, the set
$R_P$ represents all the rails rendered for viewport $P$.

For each rail $r \in R$ let $m(r) \in R$ denote the maximal rail of
$R$ containing it. Of course, $r=m(r)$ for a maximal rail.

  Let $R_P^{\max}=\{m(r) \in R : r \in R_P\}$. Rendered on the plain, the rails of $R^{\max}_P$
  cover the drawing of rails of $R_P$, because each rail is either maximal or
  is contained in a maximal rail.

\newcommand{\LemSmallTilesText}{%
  Algorithm~\ref{algo:node-rail-zl} sets $z$ in such a way that for any
  $n,i,j$ and~$T = T_{ij}^n$, we have $|V_T| \le Q_N/4$ and~$|R^{\max}_T| \le Q_R/4$.  }
\begin{lemma}\label{smalltiles}\LemSmallTilesText
\end{lemma}

 \begin{proof}

   For the first call to \ProcessLayer, we have~$n=0$. We prove the
   statement by induction.
   Consider Procedure \ProcessLayer for zoom level~$1 = 2^{n}$,
   $n=0$. The sets \insertedNodes and \maxRails are empty, and there
   is only one tile~$T = T_{0,0}^0$, which has the size of~$B$. After
   executing line~\ref{alg:candidateNodes}, the set \candidateNodes
   contains the first~$Q_N/4$ nodes in~$V$. Adding all candidate nodes
   to~$T$ would not break its node quota. All rails added to zoom
   level~$2^0$ are maximal rails. Due to the \emph{if} statement in
   line~\ref{alg:candidateIf}, candidate~$v$ is added to~$T$ in
   line~\ref{alg:candidateAdd} only if~$T$ would contain at
   most~$Q_R/4$ rails afterwards.

   Suppose the statement holds for~$n-1$, where $n \geq 1$, and let us prove
   it for~$n$. To do so, we first give a more detailed description of
   the $n+1$-th execution of \ProcessLayer than in
   Section~\ref{sec:method-description}. In line~\ref{alg:initTileMap}
   \insertedNodes is the set of nodes of layer~$L_{n-1}$. The set
   \maxRails is a set of rails maximal in~$R'$ and covering~$R'$
   completely, where~$R'$ is the set of rails of layers~$L_{0}$,
   \dots, $L_{n-1}$.
 We initialize \tileMap by
counting for each tile the number of nodes from \insertedNodes and the
number of rails from \maxRails intersected by it.
In line~\ref{alg:candidateNodes}, nodes in~$V \setminus
\insertedNodes$ are iteratively added to the set \candidateNodes,
as long as each tile intersects no more than~$Q_N/4$ nodes from
$V':=\insertedNodes \cup \candidateNodes$. Then, a constrained
triangular mesh~$M$ is built on all rails of $L_{n-1}$ and boundaries
of nodes from~$V'$.
In lines~\ref{alg:prevRailsUpd} and~\ref{alg:setPrevRailsUpd}, the
rails of layer~$L_{n-1}$ are added to layer~$L_{n}$, after possibly
being subdivided by~$M$.
 Further on the algorithm adds \candidateNodes and
the rails on the corresponding edges to $L_n$ as long as each
tile~$T_{i,j}^n$ intersects no more than~$Q_R/4$ maximal rails.

To get a better understanding of the algorithm, consider
Fig.~\ref{fig:abstract:insertion-1}. Here, we have $n=1$, and set
\candidateNodes is equal to $\{20,\dots,46\}$, representing all the
remaining nodes to assign. Thus, all node boundaries and all the rails
of the routes between nodes~$0, \dots, 19$ from $L_0$, see
Fig.~\ref{fig:abstract:triang-1}, are passed to the mesh
generator. Nodes~$20, \dots, 36$ are inserted. When adding node~$37$
and the corresponding rails the rail quota is exceeded, so no addition
happens. Finally, for $n=2$, the remaining nodes~$37, \dots, 46$ are
added to $L_2$.

Going back to the proof, us now prove the statement for~$n+1$.  For a
tile~$T$ of level~$n$, we have $l(T) = 2^{n}$ on the~$n+1$-th call to
\ProcessLayer. Tile~$T=T_{ij}^{n}$, by construction, is a subset of the
tile $T'=T_{i/2,j/2}^{n-1}$. By the induction hypothesis, $T'$
intersects at most~$Q_N/4$ nodes, and $|R^{\max}_{T'}| \le Q_R/4$.
  After the insertion in line~\ref{alg:initTileMap}, only
  \insertedNodes are on $L_{n}$, and~$T$ intersects at most~$Q_N/4$
  nodes of~$L_n$.

  Again, \candidateNodes contains a set of nodes which can be inserted
  without breaking node quota of any tile on level~$2^{n}$. For nodes
  $V':=$\insertedNodes$\cup$\candidateNodes we pass the bounding
  polygons of the current node size as well as rails of
  level~$2^{n-1}$ (\prevRails) to the mesh generator, which then
  creates a triangular mesh~$M$.

  In lines~\ref{alg:prevRailsUpd} and~\ref{alg:setPrevRailsUpd}, the
  set \prevRailsUpd contains the segments of~$M$ which are covered by
  rails in \prevRails and we do not have new maximal rails.
  Thus, directly after line~\ref{alg:setPrevRailsUpd}, $R^{\max}_{T}
  \subset R^{\max}_{T'}$ and contains at most~$Q_R/4$ elements.

  Similarly to level~1, the algorithm attempts to add \candidateNodes
  and maximal rails covering the corresponding edges to \tileMap, as long as each
  tile~$T_{i,j}^{n}$ contains no more than~$Q_R/4$ maximal rails. Because of
  the guarding \emph{if} statement in line~\ref{alg:candidateIf} we
  will not overfill the quotas for $T$ after inserting each candidate
  node~$v$ and \maxRailsV in line~\ref{alg:candidateAdd}. Thus, the
  quotas hold for~$T$ on level~$2^{n}$. \hfill$\square$

 \end{proof}

\newcommand{\ThNodeLevelsCorrectText}{On termination,
  Algorithm~\ref{algo:node-rail-zl} defines the layers in such a way
  that $|V_P| \le Q_N$ and $|R_P^{\max}| \le Q_R$.}
\begin{theorem}
\label{th:nodecorrect}
\ThNodeLevelsCorrectText
\end{theorem}
\begin{proof}
  Let $n = n(P) = \max \{ 0, \lfloor \log_2 l(P) \rfloor \}$.  By
  Lemma~\ref{algo:node-rail-zl}, each tile of $L_n$ intersects not
  more than $Q_N/4$ nodes of $L_n$, and not more than $Q_R/4$ of
  maximal rails from $\cup_{i\le n} L_i$. The tile's dimensions in
  $L_n$ are $w(B)/2^n$ and $h(B)/2^n$. If we prove that $P$ intersects
  not more than four tiles of $L_n$, then the theorem follows. For
  $n=0$ we have only one tile on $L_0$, and the theorem holds.

  Let $n > 0$. We show that $P$ ``is not greater
  than'' a tile of $L_n$, that is, $w(P) \le w(B)/2^n$ and $h(P) \le
  h(B)/2^n$, and hence, $P$ cannot intersect more than four tiles of
  $L_n$.

  By the definition of~$l(P)$, we have $l(P) \leq w(B)/w(P)$ and~$l(P)
  \leq h(B)/h(P)$.  Let us prove that $w(P) \le w(B)/2^n$. The proof
  that $h(P) \le h(B)/2^n$ is similar. Since $n > 0$, we have $n=
  \lfloor \log_2 l(P) \rfloor \le \log_2 l(P)$.  Therefore, $2^n \le
  l(P) \le w(B)/w(P)$, and $w(P) \le w(B) / 2^n$. The theorem follows.
\end{proof}

\subsection{Node labeling}
\label{app:sec:node-labeling}

We render nodes as circles and assign them text labels, both of
constant height on the screen; see Fig.~\ref{fig:b100}, right. To
guarantee readability, we aim to avoid overlaps of labels with nodes
and with each other. For a given view only a subset of visible nodes
might have visible labels.
Our approach makes use of the order of nodes in~$V$ and greedily tries
to label important nodes first.
%
Each node~$v \in V$ is assigned a value~$\ell(v) \ge z(v)$ and one
possible label position $p(v)$: left, right, above or below the node
circle. The label of $v$ is visible if and only if the current viewport
intersects the label and if it holds $Z \ge l(v)$, where
$Z$ is the current zoom level as defined in Section~\ref{sec:method-description}.

%
We iteratively consider zoom levels~$Z_i = \delta_0 \cdot
\delta^i$ for $i=0,1,2,\dots$ and some $\delta_0 > 0$, $\delta >
1$, e.g.,~$\delta_0 = 2^{-4}$ and~$\delta=2^{1/8}$. For fixed~$Z_i$,
we insert all nodes~$v$ with~$z(v) \leq Z_i$ as well as all assigned
labels into an R-tree with sizes corresponding to~$Z_i$. Next, for yet
unlabeled nodes $v$ with~$z(v) \leq Z_i$, we set~$\ell(v)=Z_i$
iteratively according to the order of nodes in~$V$, as long as no
label overlaps arise.
%
As the zoom level~$Z_i$ grows the bounding boxes of the labels and the
nodes shrink, so for some, maybe large, $i$ it is possible to find
$p(v)$ such that the label's bounding box does not overlap already
processed nodes and labels. Then, we can insert the label into the
tree. When looking for $p(v)$ we try to insert the label left, right,
above or below the node circle in this order.

This simple greedy approach produces good results in our
experiments. We plan to investigate the possibility to adapt advanced
dynamic node labeling algorithms, e.g.~\cite{bnpw-oarcd-10}, in
GraphMaps.

\subsection{Pre-rendered tiles.}
\label{app:sec:tiles}

For a tile~$T_{ij}^n$ we consider the set $t_{ij}^n = \{v \in V
\setminus L_n: \text{center}(v) \in T_{ij}^n\}$. If the number of
elements of $t_{ij}^n$ is greater than a specific threshold, 60 in our
setting, then the nodes of $t_{ij}^n$ are rendered transparently into
an image.

 GraphMaps generates tile images efficiently by using a
recursion and avoiding processing each tile of each layer. Even for a
graph with about $40000$ nodes the tile generation takes less than a
minute.

The images are loaded and unloaded dynamically by using a background
thread to keep the visualization responsive. If, while browsing, a
tile image is not found on the disk for tile $T_{ij}^n$, then we
create vector graphics elements for each node of $t_{ij}^n$. GraphMaps
works in such a way that not more than four tiles are loaded at a
time. Therefore, the number of such vector elements loaded at any
moment is not greater than $160$ and they do not hinder the
visualization.

\end{document}